\documentclass{article}[12pt]
\usepackage{a4wide}
\usepackage{amsfonts,amsmath,latexsym,times}
\newtheorem{lemma}{Lemma}
\newtheorem{proposition}{Proposition}
\newtheorem{corollary}{Corollary}
\newenvironment{proof}{\noindent {\bf Proof \ }}{\hfill ~$\dashv$}
\usepackage{color}
\newcommand{\weg}[1]{}
\newtheorem{definition}{Definition}
\newcommand{\hans}[1]{}
\renewcommand{\phi}{\varphi}
\newcommand{\m}{{\mathcal M}}
\newcommand{\imp}{\rightarrow}
\newcommand{\powerset}{\mathcal P}

\begin{document}
\title{A simple proof of the completeness of $APAL$}
\author{Philippe Balbiani\thanks{Institut de recherche en informatique de Toulouse} \ and Hans van Ditmarsch\thanks{Laboratoire lorrain de recherche en informatique et ses applications}}
\date{\today}
\maketitle
\begin{abstract}
We provide a simple proof of the completeness of arbitrary public announcement logic $APAL$. The proof is an improvement over the proof found in \cite{balbianietal:2008}.
\end{abstract}
\section{Introduction}
In \cite{balbianietal:2008} Arbitrary Public Announcement Logic ($APAL$) is presented. This is an extension of the well-known public announcement logic \cite{plaza:1989} with quantification over announcements. The logic is axiomatized, but the completeness proof may be considered rather complex. The completeness is shown by employing an infinitary axiomatization, that is then shown to be equivalent (it produces the same set of theorems) to a finitary axiomatization. The completeness proof in \cite{balbianietal:2008} contained an error in the Truth Lemma. The lemma is as follows: \weg{\begin{quote}} {\em Let $\varphi$ be a formula in ${\mathcal L}_{apal}$. Then for all maximal consistent theories
$x$ and for all finite sequences $\vec{\psi}=\psi_{1},\ldots,\psi_{k}$
of formulas in ${\mathcal L}_{apal}$ such that $\psi_{1}\in x$, $\ldots$, $
\lbrack\psi_{1}\rbrack\ldots\lbrack\psi_{k-1}\rbrack\psi_{k}\in x$:  ${\mathcal M}^{c}|{\vec{\psi}},x\models\varphi$ iff $\lbrack\psi_{1}\rbrack\ldots\lbrack\psi_{k}\rbrack\varphi\in x$.} \weg{\end{quote}} The proof is by induction on $\varphi$. The problem is that in expression ${\mathcal M}^{c}|{\vec{\psi}},x\models\varphi$, the restriction ${\mathcal M}^{c}|{\vec{\psi}}$
of the canonical model ${\mathcal M}^{c}$ cannot be assumed to exist: although we have assumed that
$\psi_{1} \in x$, $\ldots$, and that $\lbrack\psi_{1}\rbrack\ldots\lbrack
\psi_{k-1}\rbrack\psi_{k} \in x$, we did not assume that ${\mathcal M}^{c}, x \models \psi_{1}$, \dots, and that ${\mathcal M}^{c}, x \models \lbrack\psi_{1}\rbrack\ldots\lbrack\psi_{k-1}\rbrack\psi_{k}$. The latter would be needed to guarantee that existence. But the induction was only on $\phi$ and not on $\psi_{1}$, \dots, and $\lbrack\psi_{1}\rbrack\ldots\lbrack\psi_{k-1}\rbrack\psi_{k}$ as well. This error has been corrected in \cite{philippe.corrected:2014}, by an expanding the complexity measure used in the Truth Lemma to include the formulas in the sequence $\psi_{1}$, $\ldots$, $\lbrack\psi_{1}\rbrack\ldots\lbrack
\psi_{k-1}\rbrack\psi_{k}$ as well. 

Another source of confusion in \cite{balbianietal:2008}, although there was no error involved, concerned the employment of maximal consistent theories (instead of maximal consistent sets, a more common term in modal logic), and a number of properties shown for maximal consistent theories. While repairing the completeness proof, and while also considering additional properties of the canonical model, we found another completeness proof, that the reader may consider more direct and more elegant than the one in \cite{balbianietal:2008,philippe.corrected:2014}. This is presented in this work, including some further results for the canonical model.

%
%
%
%
%\section{Syntax, semantics, and axiomatization}
\section{Syntax}
Let $Atm$ be a countable set of atoms (with typical members denoted $p$, $q$, etc) and $Agt$ be a countable set of agents (with typical members denoted $a$, $b$, etc).

\begin{definition}[Language of APAL]
The set $\mathcal{L}_{apal}$ of all formulas (with typical members denoted $\phi$, $\psi$, etc) is inductively defined as follows: 
\begin{itemize}
\item $\phi::=p\mid\bot\mid\neg\phi\mid(\phi\vee\psi)\mid K_{a}\phi\mid\lbrack\phi\rbrack\psi\mid\Box\phi$.
\end{itemize}
\end{definition}
We define the other Boolean constructs as usual.
The formulas $\hat{K}_{a}\phi$, $\langle\phi\rangle\psi$ and $\Diamond\phi$ are obtained as abbreviations: $\hat{K}_{a}\phi$ for $\neg K_{a}\neg\phi$, $\langle\phi\rangle\psi$ for $\neg\lbrack\phi\rbrack\neg\psi$ and $\Diamond\phi$ for $\neg\Box\neg\phi$.
We adopt the standard rules for omission of the parentheses.
Given a formula $\phi$, the set of all subformulas of $\phi$ is denoted by $Sub(\phi)$ (an elementary inductive definition is omitted).
We will say that a formula $\phi$ is {\em $\Box$-free} iff $Sub(\phi)\cup\{\phi\}$ contains no formula of the form $\Box\psi$.
A formula $\phi$ is said to be $\lbrack\cdot\rbrack$-free iff $Sub(\phi)\cup\{\phi\}$ contains no formula of the form $\lbrack\psi\rbrack\chi$.
We will say that a formula $\phi$ is {\em epistemic} iff $\phi$ is both $\Box$-free and $\lbrack\cdot\rbrack$-free.
The set $\mathcal{L}_{pal}$ is the set of all $\Box$-free formulas.
The set $\mathcal{L}_{el}$ is the set of all epistemic formulas.

Of crucial importance in the completeness proof is a proper complexity measure on formulas. The one we need is based on a partial order $<^{Size}$ providing a weighted count of the number of symbols, and on a partial order $<_{d_\Box}$ counting the number of stacked $\Box$ operators in a formula.

\begin{definition}[Size]
The size of a formula $\phi$, in symbols $Size(\phi)$, is the non-negative integer inductively defined as follows:
\begin{itemize}
\item $Size(p)=1$,
\item $Size(\bot)=1$,
\item $Size(\neg\phi)=Size(\phi)+1$,
\item $Size(\phi\vee\psi)=Size(\phi)+Size(\psi)+1$,
\item $Size(K_{a}\phi)=Size(\phi)+1$,
\item $Size(\lbrack\phi\rbrack\psi)=Size(\phi)+3 \cdot Size(\psi)$,
\item $Size(\Box\phi)=Size(\phi)+1$.
\end{itemize}
The $\Box$-depth of a formula $\phi$, in symbols $d_{\Box}(\phi)$, is the non-negative integer inductively defined as follows:
\begin{itemize}
\item $d_{\Box}(p)=0$,
\item $d_{\Box}(\bot)=0$,
\item $d_{\Box}(\neg\phi)=d_{\Box}(\phi)$,
\item $d_{\Box}(\phi\vee\psi)=\max\{d_{\Box}(\phi),d_{\Box}(\psi)\}$,
\item $d_{\Box}(K_{a}\phi)=d_{\Box}(\phi)$,
\item $d_{\Box}(\lbrack\phi\rbrack\psi)=\max\{d_{\Box}(\phi),d_{\Box}(\psi)\}$,
\item $d_{\Box}(\Box\phi)=d_{\Box}(\phi)+1$.
\end{itemize}
We define the binary relations $<^{Size}$,  $<_{d_{\Box}}$, and $<_{d_{\Box}}^{Size}$ between formulas in the following way:
\begin{itemize}
\item $\phi<^{Size}\psi$ iff $Size(\phi)<Size(\psi)$.
\item $\phi<_{d_{\Box}}\psi$ iff $d_{\Box}(\phi)<d_{\Box}(\psi)$.
\item $\phi<_{d_{\Box}}^{Size}\psi$ iff either $d_{\Box}(\phi)<d_{\Box}(\psi)$, or $d_{\Box}(\phi)=d_{\Box}(\psi)$ and $Size(\phi)<Size(\psi)$.
\end{itemize}

\end{definition}
The next two lemmas combine a number of results on these binary relations.
Their proofs are obvious and have been omitted.
\begin{lemma} \label{gather.lemma} Let $\phi,\psi$ be formulas.
\begin{itemize}
\item \label{lem_1b}
$<^{Size}$ is a well-founded strict partial order between formulas.
%
%
%\end{lemma}
%
%\begin{lemma}
\item \label{lem_1c}
$<_{d_{\Box}}$ is a well-founded strict partial order between formulas.
%
%
% \end{lemma}
%
%
%
%
%\begin{lemma}
\item \label{lem_1}
$<_{d_{\Box}}^{Size}$ is a well-founded strict partial order between formulas.
%
%
%\end{lemma}
%
%
%\begin{lemma}
\item \label{lem_1bh}
If $\phi <^{Sub} \psi$ then $\phi <^{Size} \psi$.
\item \label{lem_hans} If $\phi <^{Sub} \psi$ then $\phi <^{Size}_{d_{\Box}} \psi$.
%
%\end{lemma}
%
%
%
% 
%\begin{lemma}
\item \label{lem_1a}
%
%
%Let $\phi$ be a formula.
If $\phi$ is epistemic, then $d_{\Box}(\phi)=0$.
\item If $\psi$ is epistemic, then $\lbrack\psi\rbrack\phi<_{d_{\Box}}^{Size}\Box\phi$.
\end{itemize}
\end{lemma}

%\begin{proof}
%The proofs of these items are obvious and have been omitted, except for the proof of the last item that is as follows. 
%Note that $d_\Box([\psi]\phi) = \max\{d_\Box(\psi),d_\Box(\phi)\}=d_\Box(\phi) < 1+ d_\Box(\phi) = d_\Box(\Box\phi)$.
%From $d_\Box([\psi]\phi) < d_\Box(\Box\phi)$ follows $[\psi]\phi <_{d_{\Box}}^{Size} \Box\phi$.
%\end{proof}

%
%
% \hans{J'ai enlev\'e les cas qui suivent directement de:  If $\phi <^{Sub} \psi$ then $\phi <^{Size}_{d_{\Box}} \psi$. (voir en haut)}
% 

\begin{lemma}\label{lem_2}
Let $\phi,\psi,\chi$ be formulas and $a\in Agt$.
\begin{enumerate}
%
%
%\item $\phi<_{d_{\Box}}^{Size}\neg\phi$,
%
%
%\item $\phi<_{d_{\Box}}^{Size}\phi\vee\psi$ and $\psi<_{d_{\Box}}^{Size}\phi\vee\psi$,
%
%
%\item $\phi<_{d_{\Box}}^{Size}K_{a}\phi$,
%
%
%\item $\phi<_{d_{\Box}}^{Size}\lbrack\phi\rbrack p$,
%
%
%\item $\phi<_{d_{\Box}}^{Size}\lbrack\phi\rbrack\bot$,
%
%
%\item $\phi<_{d_{\Box}}^{Size}\lbrack\phi\rbrack\neg\psi$,
%
%
\item $\neg\lbrack\phi\rbrack\psi<^{Size}\lbrack\phi\rbrack\neg\psi$,
%
%
%\item $\lbrack\phi\rbrack\psi<_{d_{\Box}}^{Size}\lbrack\phi\rbrack(\psi\vee\chi)$,
%
%
%\item $\lbrack\phi\rbrack\chi<_{d_{\Box}}^{Size}\lbrack\phi\rbrack(\psi\vee\chi)$,
%
%
%\item $\phi<_{d_{\Box}}^{Size}\lbrack\phi\rbrack K_{a}\psi$,
%
%
\item $K_{a}\lbrack\phi\rbrack\psi<^{Size}\lbrack\phi\rbrack K_{a}\psi$,
\item $\lbrack\neg\lbrack\phi\rbrack\neg\psi\rbrack\chi<^{Size}\lbrack\phi\rbrack\lbrack\psi\rbrack\chi$.
%
%
% \item if $\chi$ is epistemic, then  $\lbrack\phi\rbrack\lbrack\chi\rbrack\psi<_{d_{\Box}}^{Size}\lbrack\phi\rbrack\Box\psi$.
%
%
%
%
\end{enumerate}
\end{lemma}

The relation $<^{Size}$ has been tailored in order to ensure exactly the properties of Lemma~\ref{lem_2}.
Without the curious factor $3$ in $Size(\lbrack\phi\rbrack\psi)=Size(\phi)+3 \cdot Size(\psi)$ these properties would not hold.
Given the previous lemmas, we can now list all the cases later used in the Truth Lemma.
\begin{corollary} In cases $(*)$ and $(**)$, $\phi$ is epistemic.
\[ \begin{array}{llll|lll|lll}
&\phi &<_{d_{\Box}}^{Size}& \neg\phi &        \phi &<_{d_{\Box}}^{Size}& [\phi]p               & \phi &<_{d_{\Box}}^{Size}& [\phi]K_{a}\psi \\
&\phi &<_{d_{\Box}}^{Size}& \phi\vee\psi &    \phi &<_{d_{\Box}}^{Size}& [\phi]\bot                   & K_{a} [\phi]\psi &<_{d_{\Box}}^{Size}& [\phi] K_{a} \psi \\
&\psi &<_{d_{\Box}}^{Size}& \phi\vee\psi &    \phi &<_{d_{\Box}}^{Size}& [\phi]\neg\psi             & [\neg [\phi]\neg \psi]\chi &<_{d_{\Box}}^{Size}& [\phi][\psi]\chi \\
&\phi &<_{d_{\Box}}^{Size}& K_{a}\phi &    [\phi]\psi &<_{d_{\Box}}^{Size}& [\phi]\neg\psi       & [\chi][\phi]\psi &<_{d_{\Box}}^{Size}& [\chi]\Box \psi \ \hspace{1cm} (**) \\
(*)&[\phi]\psi &<_{d_{\Box}}^{Size}& \Box\psi & [\phi]\psi &<_{d_{\Box}}^{Size}& [\phi](\psi\vee\chi) & \\
 & &&&                      [\phi]\chi &<_{d_{\Box}}^{Size}& [\phi](\psi\vee\chi) & 
\end{array} \]
\end{corollary}
\begin{definition}[Necessity form]
Now, let us consider a new atom denoted $\sharp$.
The set $NF$ of {\em necessity forms} (with typical members denoted $\xi(\sharp)$, $\xi^{\prime}(\sharp)$, etc) is inductively defined as follows---where $\phi$ is a formula.
\begin{itemize}
\item $\xi(\sharp)::=\sharp\mid\phi\rightarrow\xi(\sharp)\mid K_{a}\xi(\sharp)\mid\lbrack\phi\rbrack\xi(\sharp)$.
\end{itemize}
\end{definition}
\section{Semantics}

We introduce the structures and give a semantics for the logical language on these structures. The material in this section (as also the logical language in the previous section, and the axiomatization in the next section) is as in \cite{balbianietal:2008}.

\begin{definition}[Model] \label{def.model}
A model $\m = (W, R, V )$ consists of a nonempty {\em domain} $W$, an {\em accessibility function} $R: Agt \imp \powerset(W \times W)$ associating to each $a\in Agt$ an equivalence relation $R(a)$ on $W$, and a {\em valuation function} $V: Atm \imp \powerset(W)$ --- where $V(p)$ denotes the valuation of atom $p$.
For $R(a)$, we write $R_{a}$.
\end{definition}

\begin{definition}[Semantics]
Assume a model $\m = (W, R, V )$. We inductively define the truth set $\parallel \phi \parallel^{\m}$.
\[ \begin{array}{lcl}
w \in \ \parallel p \parallel^\m &\mbox{iff} & w \in V(p) \\ 
w \in \ \parallel \neg \phi\parallel^\m  &\mbox{iff} & w \not\in \ \parallel \phi \parallel^\m \\ 
w \in \ \parallel  \phi \vee \psi\parallel^\m  &\mbox{iff} &  w \in \ \parallel \phi  \parallel^\m \text{ or } w \in \ \parallel \psi  \parallel^\m  \\  
w \in \ \parallel K_{a} \phi\parallel^\m  &\mbox{iff} &  \text{for all } v, R_{a}(w,v) \text{ implies } v \in \ \parallel \phi \parallel^\m \\  
w \in \ \parallel  [\phi] \psi \parallel^\m &\mbox{iff} &  w \in \ \parallel \phi  \parallel^\m \text{ implies } w \in \ \parallel \psi \parallel ^{\m^\phi} \\
w \in \ \parallel  \Box \psi \parallel^\m &\mbox{iff} & \text{for all epistemic } \phi, w \in \ \parallel [\phi]\psi \parallel^\m 
\end{array} \] 
where model $\m^\phi = (W', R', V')$ is such that \[\begin{array}{lll} W' &=& \parallel \phi \parallel^\m, \\ R'_a &=& R_a \cap (\parallel\phi \parallel^\m \times \parallel\phi \parallel^\m), \\ V'(p) &=& V(p) \cap \parallel \phi \parallel^\m. \end{array} \]
\end{definition} 
\section{Axiomatization}
An axiomatic system consists of a collection of axioms and a collection of inference rules.
Let us consider the following axiomatic system:
\begin{definition}[Axiomatization $APAL$]
\begin{description}
\item[$(A0)$] all instantiations of propositional tautologies,
\item[$(A1)$] $K_{a}(\phi\rightarrow\psi)\rightarrow(K_{a}\phi\rightarrow K_{a}\psi)$,
\item[$(A2)$] $\lbrack\phi\rbrack(\psi\rightarrow\chi)\rightarrow(\lbrack\phi\rbrack\psi\rightarrow\lbrack\phi\rbrack\chi)$,
\item[$(A3)$] $\Box(\phi\rightarrow\psi)\rightarrow(\Box\phi\rightarrow\Box\psi)$,
\item[$(A4)$] $K_{a}\phi\rightarrow\phi$,
\item[$(A5)$] $K_{a}\phi\rightarrow K_{a}K_{a}\phi$,
\item[$(A6)$] $\phi\rightarrow K_{a}\hat{K}_{a}\phi$,
\item[$(A7)$] $\lbrack\phi\rbrack p\leftrightarrow(\phi\rightarrow p)$,
\item[$(A8)$] $\lbrack\phi\rbrack\bot\leftrightarrow\neg\phi$,
\item[$(A9)$] $\lbrack\phi\rbrack\neg\psi\leftrightarrow(\phi\rightarrow\neg\lbrack\phi\rbrack\psi)$,
\item[$(A10)$] $\lbrack\phi\rbrack(\psi\vee\chi)\leftrightarrow\lbrack\phi\rbrack\psi\vee\lbrack\phi\rbrack\chi$,
\item[$(A11)$] $\lbrack\phi\rbrack K_{a}\psi\leftrightarrow(\phi\rightarrow K_{a}\lbrack\phi\rbrack\psi)$,
\item[$(A12)$] $\lbrack\phi\rbrack\lbrack\psi\rbrack\chi\leftrightarrow\lbrack\langle\phi\rangle\psi\rbrack\chi$,
\item[$(A13)$] if $\psi$ is epistemic, then $\Box\phi\rightarrow\lbrack\psi\rbrack\phi$,
\item[$(R0)$] $(\{\phi,\phi\rightarrow\psi\},\psi)$,
\item[$(R1)$] $(\{\phi\},K_{a}\phi)$,
\item[$(R2)$] $(\{\phi\},\lbrack\psi\rbrack\phi)$,
\item[$(R3)$] $(\{\phi\},\Box\phi)$, %derivable
\item[$(R4)$] $(\{\xi(\lbrack\psi\rbrack\phi)$: $\psi$ is epistemic$\},\xi(\Box\phi))$.
\end{description}

\bigskip

\noindent
Let $APAL$ be the least subset of $\mathcal{L}_{apal}$ containing $(A0)$--$(A13)$ and closed under $(R0)$--$(R4)$. An element of $APAL$ is called a {\em theorem}.
\end{definition}
In \cite{balbianietal:2008} other (finitary) axiomatizations are also given, that are then shown to be equivalent to $APAL$ (they define the same set of theorems as $APAL$). For the completeness proof, we have chosen the most convenient form, with the infinitary rule $(R4)$.  Some of the axioms and rules in the axiomatization $APAL$ are derivable from the other axioms and rules, again, see \cite{balbianietal:2008} for details. It concerns the following rules and axioms (where $\bot$ should be seen as the abbreviation of $p \wedge \neg p$):
\begin{description}
\item[$(A3)$] $\Box(\phi\rightarrow\psi)\rightarrow(\Box\phi\rightarrow\Box\psi)$;
\item[$(A8)$] $\lbrack\phi\rbrack\bot\leftrightarrow\neg\phi$; 
%\item[$(A10)$]  $\lbrack\phi\rbrack(\psi\vee\chi)\leftrightarrow\lbrack\phi\rbrack\psi\vee\lbrack\phi\rbrack\chi$;
\item[$(R3)$] $(\{\phi\},\Box\phi)$. %derivable
\end{description}
\section{Canonical model}
\begin{definition}[Theory]
A set $x$ of formulas is called a theory iff it satisfies the following conditions:
\begin{itemize}
\item $x$ contains $APAL$,
\item $x$ is closed under $(R0)$ and $(R4)$.
\end{itemize}
A theory $x$ is said to be {\em consistent} iff $\bot\not\in x$.
A set $x$ of formulas is {\em maximal} iff for all formulas $\phi$, $\phi\in x$ or $\neg\phi\in x$. 
\end{definition}
Obviously, the smallest theory is $APAL$ whereas the largest theory is $\mathcal{L}_{apal}$.
The only inconsistent theory is $\mathcal{L}_{apal}$.
The reader may easily verify that a theory $x$ is consistent iff for all formulas $\phi$, $\phi\not\in x$ or $\neg\phi\not\in x$.
Moreover, for all maximal consistent theories $x$,
\begin{itemize}
\item $\bot\not\in x$,
\item $\neg\phi\in x$ iff $\phi\not\in x$,
\item $(\phi\vee\psi)\in x$ iff $\phi\in x$ or $\psi\in x$.
\end{itemize}
Theories are closed under $(R0)$ and $(R4)$ but not under the derivation rules $(R1)$, $(R2)$, and $(R3)$ for a specific reason. Obviously, by definition, all derivation rules preserve theorems. Semantically, we could say that they all preserve validities. Now, unlike $(R1)$, $(R2)$, and $(R3)$, the derivation rules $(R0)$ and $(R4)$ also preserve truths. That is the reason!
In the setting of our axiomatization based on the infinitary rule (R4), we will say that a set $x$ of formulas is consistent iff there exists a consistent theory $y$ such that $x\subseteq y$.
Obviously, maximal consistent theories are maximal consistent sets of formulas. Under the given definition of consistency for sets of formulas, maximal consistent sets of formulas are also maximal consistent theories.
% [{\bf But under other definitions of consistency, such as derivability in the finitary alternative axiomatization in \cite{balbianietal:2008}, this identification cannot be made. More details to be given here.}] \hans{Mai j'ai oubli\'e ce que tu m'as expliqu\'e en d\'etail \`a Toulouse. Il faut ajouter une text ici!}
%
%
\begin{definition}
For all formulas $\phi$ and for all $a\in Agt$, let \[\begin{array}{lll} x+\phi &=&\{\psi : \phi\rightarrow\psi\in x\}, \\ K_{a}x &=& \{\phi: K_{a}\phi\in x\} \\ \lbrack\phi\rbrack x &=& \{\psi: \lbrack\phi\rbrack\psi\in x\}.  \end{array} \]
\end{definition}
The proofs of the following lemmas can be found in~\cite{balbianietal:2008} (Lemmas~$4.11$ and~$4.12$).
\begin{lemma}\label{lem_7}
Let $\phi$ be a formula and $a\in Agt$.
For all theories $x$,
\begin{itemize}
\item $x+\phi$ is a theory containing $x$ and $\phi$,
\item $\lbrack\phi\rbrack x$ is a theory,
\item $K_{a}x$ is a theory.
\end{itemize}
\end{lemma}
\begin{lemma}\label{lem_7bis}
Let $\phi$ be a formula.
For all theories $x$, $x+\phi$ is consistent iff $\neg\phi\not\in x$.
\end{lemma}
\begin{lemma}\label{lem_9}
Each consistent theory can be extended to a maximal consistent theory.
\end{lemma}
The proof of the next lemma uses axioms $(A4)$--$(A6)$.
\begin{lemma}\label{equivalence}
Let $a\in Agt$.
For all maximal consistent theories $x,y,z$,
\begin{itemize}
\item $K_{a}x\subseteq x$,
\item if $K_{a}x\subseteq y$ and $K_{a}y\subseteq z$, then $K_{a}x\subseteq z$,
\item if $K_{a}x\subseteq y$, then $K_{a}y\subseteq x$.
\end{itemize}
\end{lemma}
Next lemma is usually called ``Diamond Lemma''.
Its proof is very classical and uses Lemmas~\ref{lem_7}, \ref{lem_7bis} and~\ref{lem_9}.
\begin{lemma}\label{diamond_lemma}
Let $\phi$ be a formula and $a\in Agt$.
For all theories $x$, if $K_{a}\phi\not\in x$, then there exists a maximal consistent theory $y$ such that $K_{a}x\subseteq y$ and $\phi\not\in y$.
\end{lemma}
%
%
%\begin{proof}
%
%   %hans 19h00 samedi
%Suppose $K_{a}\phi\not\in x$.
%Hence, $\phi\not\in K_{a}x$ and $\neg\neg\phi\not\in K_{a}x$.
%By Lemmas~\ref{lem_7} and~\ref{lem_7bis}, $K_{a}x+\neg\phi$ is a consistent theory.
%By Lemma~\ref{lem_9}, there exists a maximal consistent theory $y$ such that $K_{a}x+\neg\phi\subseteq y$.
%Therefore, $K_{a}x\subseteq y$ and $\phi\not\in y$.
%
%
%\end{proof}
%
%
The next three lemmas were not found in \cite{balbianietal:2008}.
 %We find they increase our understanding of  $APAL$. \hans{vrai? Il y a d'autres lemmes nouveaux?}
\begin{lemma}\label{lem_mcs}
Let $\phi$ be a formula.
For all maximal consistent theories $x$, if $\phi\in x$, then $\lbrack\phi\rbrack x$ is a maximal consistent theory.
\end{lemma}
\begin{proof}
Suppose $\phi\in x$.
If $\lbrack\phi\rbrack x$ is not consistent, then $\bot\in\lbrack\phi\rbrack x$.
Hence, $\lbrack\phi\rbrack\bot\in x$.
Thus, $\neg\phi\in x$.
Since $x$ is consistent, $\phi\not\in x$: a contradiction.
If $\lbrack\phi\rbrack x$ is not maximal, then there exists a formula $\psi$ such that $\psi\not\in\lbrack\phi\rbrack x$ and $\neg\psi\not\in\lbrack\phi\rbrack x$.
Therefore, $\lbrack\phi\rbrack\psi\not\in x$ and $\lbrack\phi\rbrack\neg\psi\not\in x$.
Since $x$ is maximal, $\neg\lbrack\phi\rbrack\psi\in x$ and $\neg\lbrack\phi\rbrack\neg\psi\in x$.
Consequently, $\neg(\lbrack\phi\rbrack\psi\vee\lbrack\phi\rbrack\neg\psi)\in x$.
Hence, using $(A10)$, $\neg\lbrack\phi\rbrack(\psi\vee\neg\psi)\in x$.
Since $x$ is consistent, $\lbrack\phi\rbrack(\psi\vee\neg\psi)\not\in x$.
Since $\psi\vee\neg\psi\in APAL$, $\lbrack\phi\rbrack(\psi\vee\neg\psi)\in APAL$.
Thus, $\lbrack\phi\rbrack(\psi\vee\neg\psi)\in x$: a contradiction.
\end{proof}
\begin{lemma}\label{lem_in_diamond}
Let $\phi,\psi$ be formulas.
For all maximal consistent theories $x$, $\langle\phi\rangle\psi\in x$ iff $\phi\in x$ and $\psi\in\lbrack\phi\rbrack x$.
\end{lemma}
\begin{proof}
$(\Rightarrow)$ Suppose $\langle\phi\rangle\psi\in x$.
Hence, $\langle\phi\rangle\top\in x$.
Thus, using $(A8)$, $\phi\in x$.
By Lemma~\ref{lem_mcs}, $\lbrack\phi\rbrack x$ is a maximal consistent theory.
Suppose $\psi\not\in\lbrack\phi\rbrack x$.
Since $\lbrack\phi\rbrack x$ is maximal, $\neg\psi\in\lbrack\phi\rbrack x$.
Therefore, $\lbrack\phi\rbrack\neg\psi\in x$.
Consequently, $\neg\langle\phi\rangle\psi\in x$.
Since $x$ is consistent, $\langle\phi\rangle\psi\not\in x$: a contradiction.
\\
$(\Leftarrow)$ Suppose $\phi\in x$ and $\psi\in\lbrack\phi\rbrack x$.
By Lemma~\ref{lem_mcs}, $\lbrack\phi\rbrack x$ is a maximal consistent theory.
Suppose $\langle\phi\rangle\psi\not\in x$.
Since $x$ is maximal, $\neg\langle\phi\rangle\psi\in x$.
Hence, $\lbrack\phi\rbrack\neg\psi\in x$.
Thus, $\neg\psi\in\lbrack\phi\rbrack x$.
Since $\lbrack\phi\rbrack x$ is consistent, $\psi\not\in\lbrack\phi\rbrack x$: a contradiction.
\end{proof}
\begin{lemma}\label{commutatif}
Let $\phi$ be a formula and $a\in Agt$.
For all theories $x$, if $\phi\in x$, then $K_{a}\lbrack\phi\rbrack x=\lbrack\phi\rbrack K_{a}x$.
\end{lemma}
\begin{proof}
Suppose $\phi\in x$.
For all formulas $\psi$, the reader may easily verify that the following conditions are equivalent:
\begin{enumerate}
\item $\psi\in K_{a}\lbrack\phi\rbrack x$,
\item $K_{a}\psi\in\lbrack\phi\rbrack x$,
\item $\lbrack\phi\rbrack K_{a}\psi\in x$,
\item $\phi\rightarrow K_{a}\lbrack\phi\rbrack\psi\in x$,
\item $K_{a}\lbrack\phi\rbrack\psi\in x$,
\item $\lbrack\phi\rbrack\psi\in K_{a}x$,
\item $\psi\in\lbrack\phi\rbrack K_{a}x$.
\end{enumerate}
\end{proof}
\begin{definition}[Canonical model]
The {\em canonical model} $\mathcal{M}^{c}=(W^{c},R^{c},V^{c})$ is defined as follows:
\begin{itemize}
\item $W^{c}$ is the set of all maximal consistent theories;
\item $R^{c}$ is the function assigning to each agent $a$ the binary relation $R^{c}_a$ on $W^{c}$ defined as \[ xR^{c}_a y \text{ iff } K_{a}x\subseteq y; \]
\item $V^{c}$ is the function assigning to each atom $p$ the subset $V^{c}(p)$ of $W^{c}$ defined as \[ x\in V^{c}(p) \text{ iff } p\in x. \]
\end{itemize}
\end{definition}
It will be clear that the canonical model is a model according to Definition~\ref{def.model}.  By Lemma~\ref{lem_9}, $W^{c}$ is a non-empty set, and by Lemma~\ref{equivalence} the binary relation $R^{c}(a)$ is an equivalence relation on $W^{c}$ for each agent $a$.
\section{Completeness}
%
%\hans{Ajout\'e!} \noindent This completeness proof is different from and simpler than the argument in \cite{balbianietal:2008}. The insight that makes Truth Lemma \ref{induction_bb} work is that it is not merely by induction on the structure of $\phi$ but with a subinduction for the case $\phi = [\psi]\chi$, on the structure of $\chi$. Both the main induction and the subinduction use the complexity measure $<^{Size}_{d_\Box}$, but the real power of $<^{Size}_{d_\Box}$ emerges from the subinduction part (namely by using equivalences that swap the order of the announcement $\psi$ and the postcondition $\chi$). We recall that \cite[Truth~Lemma~4.13]{balbianietal:2008} uses finite sequences of announcements. That might be called less elegant.
The main result of this Section is the proof of $APAL$'s Truth Lemma (Lemma~\ref{equ_sem}).
This proof is different from and simpler than the proof presented in \cite{balbianietal:2008}.

\begin{definition}
Let $\phi$ be a formula. Condition $P(\phi)$ is defined as follows.  \begin{quote} For all maximal consistent theories $x$,  $\phi\in x$ iff $x\in \ \parallel\phi\parallel^{\mathcal{M}^{c}}$. \end{quote}
Condition $H(\phi)$ is defined as follows. \begin{quote} For all formulas $\psi$, if $\psi<_{d_{\Box}}^{Size}\phi$, then $P(\psi)$. \end{quote}
\end{definition}

Our new proof of $APAL$'s Truth Lemma is done by using an $<_{d_{\Box}}^{Size}$-induction on formulas.
More precisely, we will demonstrate that

%Lemma \ref{induction_bb} is formulated employing two auxiliary notions, $P$, and $H$, first defined below.
%
%Let $P$ be the set of all formulas such that the following condition holds:
%
%
%\begin{itemize}
%
%
%\item for all maximal consistent theories $x$,
%
%
%\begin{itemize}
%
%
%\item $\phi\in x$ iff $x\in \ \parallel\phi\parallel^{\mathcal{M}^{c}}$.
%
%
%\end{itemize}
%
%
%\end{itemize}
%
%
\begin{lemma}\label{induction_bb}
For all formulas $\phi$, if $H(\phi)$, then $P(\phi)$.
\end{lemma}
\begin{proof}
Suppose $H(\phi)$.
Let $x$ be a maximal consistent theory.
We consider the following $13$ cases.

\medskip \noindent 
{\bf Case $\phi=p$.}
\hans{Chang\'e} $P(p)$ holds, as $p \in x$ iff $x \in \ \parallel p \parallel^{\mathcal{M}^{c}}$, by the definition of the canonical model and the semantics of propositional atoms. %Left to the reader.

\medskip  \noindent 
{\bf Case $\phi=\bot$.}
$P(\bot)$ holds, as $\bot \not\in x$ and $x\not\in \ \parallel\bot\parallel^{\mathcal{M}^{c}}$, by the definition of the canonical model and the semantics of $\bot$. %Left to the reader.

\medskip  \noindent 
{\bf Case $\phi=\neg\psi$.}
The reader may easily verify that the following conditions are equivalent.
The induction using $<^{Size}_{d_\Box}$ is used between step 2.\ and step 3. A similar inductive argument is also used in all following cases. 
\begin{enumerate}
\item $\neg\psi\in x$,
\item $\psi\not\in x$,
\item $x\not\in \ \parallel\psi\parallel^{\mathcal{M}^{c}}$,
\item $x\in \ \parallel\neg\psi\parallel^{\mathcal{M}^{c}}$.
\end{enumerate}
Hence, $\neg\psi\in x$ iff $x\in \ \parallel\neg\psi\parallel^{\mathcal{M}^{c}}$.

\medskip \noindent 
{\bf Case $\phi=\psi\vee\chi$.}
The reader may easily verify that the following conditions are equivalent:
\begin{enumerate}
\item $\psi\vee\chi\in x$,
\item $\psi\in x$, or $\chi\in x$,
\item $x\in \ \parallel\psi\parallel^{\mathcal{M}^{c}}$, or $x\in \ \parallel\chi\parallel^{\mathcal{M}^{c}}$,
\item $x\in \ \parallel\psi\vee\chi\parallel^{\mathcal{M}^{c}}$.
\end{enumerate}
Hence, $\psi\vee\chi\in x$ iff $x\in \ \parallel\psi\vee\chi\parallel^{\mathcal{M}^{c}}$.

\medskip \noindent 
{\bf Case $\phi=K_{a}\psi$.} The reader may easily verify that the following conditions are equivalent.
%The implication from step 1.\ to step 2.\ is by the definition of the canonical model.
The implication from step 2.\ to step 1.\ is by Lemma \ref{diamond_lemma}.
\begin{enumerate}
\item $K_{a}\psi\in x$,
\item for all maximal consistent theories $y$, if $K_{a}x\subseteq y$, then $\psi\in y$,
\item for all maximal consistent theories $y$, if $xR^{c}(a)y$, then $y\in \ \parallel\psi\parallel^{\mathcal{M}^{c}}$,
\item $x\in \ \parallel K_{a}\psi\parallel^{\mathcal{M}^{c}}$.
\end{enumerate}
Hence, $K_{a}\psi\in x$ iff $x\in \ \parallel K_{a}\psi\parallel^{\mathcal{M}^{c}}$.

\medskip

%{\em We now start the case $\phi = [\psi]\chi$, but listing all inductive cases for $\chi$.}

\medskip \noindent 
{\bf Case $\phi=\lbrack\psi\rbrack p$.}
The reader may easily verify that the following conditions are equivalent. Between step 1.\ and step 2., use axiom $(A7)$ $[\psi]p \leftrightarrow (\psi \rightarrow p)$, so that  $\lbrack\psi\rbrack p\in x$ iff  $\psi \rightarrow p\in x$ (similar justifications apply in the other cases of form $[\psi]\chi$).
\begin{enumerate}
\item $\lbrack\psi\rbrack p\in x$,
\item $\psi\not\in x$, or $p\in x$,
\item $x\not\in \ \parallel\psi\parallel^{\mathcal{M}^{c}}$, or $x\in \ \parallel p\parallel^{\mathcal{M}^{c}}$,
\item $x\in \ \parallel\lbrack\psi\rbrack p\parallel^{\mathcal{M}^{c}}$.
\end{enumerate}
Hence, $\lbrack\psi\rbrack p\in x$ iff $x\in \ \parallel\lbrack\psi\rbrack p\parallel^{\mathcal{M}^{c}}$.

\medskip \noindent 
{\bf Case $\phi=\lbrack\psi\rbrack\bot$.}
The reader may easily verify that the following conditions are equivalent:
\begin{enumerate}
\item $\lbrack\psi\rbrack\bot\in x$,
\item $\psi\not\in x$,
\item $x\not\in \ \parallel\psi\parallel^{\mathcal{M}^{c}}$,
\item $x\in \ \parallel\lbrack\psi\rbrack\bot\parallel^{\mathcal{M}^{c}}$.
\end{enumerate}
Hence, $\lbrack\psi\rbrack\bot\in x$ iff $x\in \ \parallel\lbrack\psi\rbrack\bot\parallel^{\mathcal{M}^{c}}$.

\medskip \noindent 
{\bf Case $\phi=\lbrack\psi\rbrack\neg\chi$.}
The reader may easily verify that the following conditions are equivalent. In the crucial equivalence between step 2.\ and 3.\ we use that $\neg\lbrack\psi\rbrack\chi <^{Size}_{d_\Box} \lbrack\psi\rbrack\neg\chi$, a consequence of Lemma \ref{lem_2} (the $d_\Box$ depth is the same for both formulas). 
\begin{enumerate}
\item $\lbrack\psi\rbrack\neg\chi\in x$,
\item $\psi\not\in x$, or $\neg\lbrack\psi\rbrack\chi\in x$,
\item $x\not\in \ \parallel\psi\parallel^{\mathcal{M}^{c}}$, or $x\in \ \parallel\neg\lbrack\psi\rbrack\chi\parallel^{\mathcal{M}^{c}}$,
\item $x\in \ \parallel\lbrack\psi\rbrack\neg\chi\parallel^{\mathcal{M}^{c}}$.
\end{enumerate}
Hence, $\lbrack\psi\rbrack\neg\chi\in x$ iff $x\in \ \parallel\lbrack\psi\rbrack\neg\chi\parallel^{\mathcal{M}^{c}}$.

\medskip \noindent 
{\bf Case $\phi=\lbrack\psi\rbrack(\chi\vee\theta)$.}
The reader may easily verify that the following conditions are equivalent:
\begin{enumerate}
\item $\lbrack\psi\rbrack(\chi\vee\theta)\in x$,
\item $\lbrack\psi\rbrack\chi\in x$, or $\lbrack\psi\rbrack\theta\in x$,
\item $x\in \ \parallel\lbrack\psi\rbrack\chi\parallel^{\mathcal{M}^{c}}$, or $x\in \ \parallel\lbrack\psi\rbrack\theta\parallel^{\mathcal{M}^{c}}$
\item $x\in \ \parallel\lbrack\psi\rbrack(\chi\vee\theta)\parallel^{\mathcal{M}^{c}}$.
\end{enumerate}
Hence, $\lbrack\psi\rbrack(\chi\vee\theta)\in x$ iff $x\in \ \parallel\lbrack\psi\rbrack(\chi\vee\theta)\parallel^{\mathcal{M}^{c}}$.

\medskip \noindent 
{\bf Case $\phi=\lbrack\psi\rbrack K_{a}\chi$.}
The reader may easily verify that the following conditions are equivalent (again, a crucial step is between 2.\ and 3.\, where we can use induction on $K_{a}\lbrack\psi\rbrack\chi$ because of Lemma \ref{lem_2}): 
\begin{enumerate}
\item $\lbrack\psi\rbrack K_{a}\chi\in x$,
\item $\psi\not\in x$, or $K_{a}\lbrack\psi\rbrack\chi\in x$,
\item $x\not\in \ \parallel\psi\parallel^{\mathcal{M}^{c}}$, or $x\in \ \parallel K_{a}\lbrack\psi\rbrack\chi\parallel^{\mathcal{M}^{c}}$,
\item $x\in \ \parallel\lbrack\psi\rbrack K_{a}\chi\parallel^{\mathcal{M}^{c}}$.
\end{enumerate}
Hence, $\lbrack\psi\rbrack K_{a}\chi\in x$ iff $x\in \ \parallel\lbrack\psi\rbrack K_{a}\chi\parallel^{\mathcal{M}^{c}}$.

\medskip \noindent 
{\bf Case $\phi=\lbrack\psi\rbrack\lbrack\chi\rbrack\theta$.}
The reader may easily verify that the following conditions are equivalent (and once more, a crucial step is between 2.\ and 3.\, where we use Lemma \ref{lem_2}): 
\begin{enumerate}
\item $\lbrack\psi\rbrack\lbrack\chi\rbrack\theta\in x$,
\item $\lbrack\neg\lbrack\psi\rbrack\neg\chi\rbrack\theta\in x$,
\item $x\in \ \parallel\lbrack\neg\lbrack\psi\rbrack\neg\chi\rbrack\theta\parallel^{\mathcal{M}^{c}}$,
\item $x\in \ \parallel\lbrack\psi\rbrack\lbrack\chi\rbrack\theta\parallel^{\mathcal{M}^{c}}$.
\end{enumerate}
Hence, $\lbrack\psi\rbrack\lbrack\chi\rbrack\theta\in x$ iff $x\in \ \parallel\lbrack\psi\rbrack\lbrack\chi\rbrack\theta\parallel^{\mathcal{M}^{c}}$.

\medskip \noindent 
{\bf Case $\phi=\lbrack\psi\rbrack\Box\chi$.}
The reader may easily verify that the following conditions are equivalent. Between 1.\ and 2., we use derivation rule $(R4)$ on the necessity form $\lbrack\psi\rbrack\lbrack\theta\rbrack\chi$ and closure of maximal consistent sets under $(R4)$. Between step 2.\ and step 3.\, we use the complexity measure $<^{Size}_{d_\Box}$, where we now simply observe that $\lbrack\psi\rbrack\Box\chi$ contains one $\Box$ less than $\lbrack\psi\rbrack\lbrack\theta\rbrack\chi$. Between step 3.\ and step 4., we use the semantics of arbitrary announcements $\Box$ and of announcements $[\psi]$: we note that  $x\in \  \parallel\lbrack\psi\rbrack\lbrack\theta\rbrack\chi\parallel^{\mathcal{M}^{c}}$ is by the semantics equivalent to: $x\in \  \parallel\psi\parallel^{\mathcal{M}^{c}}$ implies $x\in \  \parallel\lbrack\theta\rbrack\chi\parallel^{({\mathcal{M}^{c}})^\psi}$.
\begin{enumerate}
\item $\lbrack\psi\rbrack\Box\chi\in x$,
\item for all epistemic formulas $\theta$, $\lbrack\psi\rbrack\lbrack\theta\rbrack\chi\in x$,
\item for all epistemic formulas $\theta$, $x\in \ \parallel\lbrack\psi\rbrack\lbrack\theta\rbrack\chi\parallel^{\mathcal{M}^{c}}$,
\item $x\in \ \parallel \lbrack\psi\rbrack\Box\chi\parallel^{\mathcal{M}^{c}}$.
\end{enumerate}
Hence, $\lbrack\psi\rbrack\Box\chi\in x$ iff $x\in \ \parallel \lbrack\psi\rbrack\Box\chi\parallel^{\mathcal{M}^{c}}$.

\medskip \noindent 
{\bf Case $\phi=\Box\psi$.}
The reader may easily verify that the following conditions are equivalent.
The equivalence between step 2.\ and step 3.\ follows from the fact that for all epistemic formulas $\chi$, $\lbrack\chi\rbrack\psi<^{Size}_{d_{\Box}}\Box\psi$.
\begin{enumerate}
\item $\Box\psi\in x$,
\item for all epistemic formulas $\chi$, $\lbrack\chi\rbrack\psi\in x$,
\item for all epistemic formulas $\chi$, $x\in \ \parallel\lbrack\chi\rbrack\psi\parallel^{\mathcal{M}^{c}}$,
\item $x\in \ \parallel\Box\psi\parallel^{\mathcal{M}^{c}}$.
\end{enumerate}
Hence, $\Box\psi\in x$ iff $x\in \ \parallel\Box\psi\parallel^{\mathcal{M}^{c}}$.
\end{proof}
\begin{lemma}[Truth Lemma] \label{equ_sem}
Let $\phi$ be a formula.
For all maximal consistent theories $x$,
\begin{itemize}
\item $\phi\in x$ iff $x\in \ \parallel\phi\parallel^{\mathcal{M}^{c}}$.
\end{itemize}
\end{lemma}
\begin{proof}
By Lemma~\ref{induction_bb}, using the well-foundedness of the strict partial order $<^{Size}_{d_{\Box}}$ between formulas.
\end{proof}

\medskip

\noindent Now, we are ready to prove the completeness of $APAL$.
\begin{proposition}\label{pro_complete}
For all formulas $\phi$, if $\phi$ is valid, then $\phi\in APAL$.
\end{proposition}
\begin{proof}
Suppose $\phi$ is valid and $\phi\not\in APAL$.
By Lemmas~\ref{lem_7}, \ref{lem_7bis} and~\ref{lem_9}, there exists a maximal consistent theory $x$ containing $\neg\phi$.
By Lemma~\ref{equ_sem}, $x\in \ \parallel\neg\phi\parallel^{\mathcal{M}^{c}}$.
Thus, $x\not\in \ \parallel\phi\parallel^{\mathcal{M}^{c}}$.
Therefore, $\parallel\phi\parallel^{\mathcal{M}^{c}}\not=W^{c}$.
Consequently, $\phi$ is not valid: a contradiction.
\end{proof}
\section{Conclusion}
We have provided an alternative, simpler, completeness proof for the logic $APAL$. The proof is considered simpler, because in the crucial Truth Lemma we do not need to take finite sequences of announcements along. Instead, it can proceed by $<^{Size}_{d_{\Box}}$-induction on formulas.
%structure, with a crucial subinduction on the postcondition of the case of public announcement in that induction.
We consider this result useful, as the completeness proofs of various other logics employing arbitrary announcements or other forms of quantifiying over announcements may thus also be simplified, and as it may encourage the developments of novel logics with quantification over announcements. We acknowledge useful discussions on the completeness of $APAL$ with Jie Fan, Wiebe van der Hoek, and Barteld Kooi.

\bibliographystyle{plain}
\bibliography{biblio2014}

\begin{thebibliography}{1}

\bibitem{philippe.corrected:2014}
P.~Balbiani.
\newblock A new proof of completeness for {APAL}.
\newblock Manuscript under submission, 2014.

\bibitem{balbianietal:2008}
P.~Balbiani, A.~Baltag, H.~van Ditmarsch, A.~Herzig, T.~Hoshi, and T.~De Lima.
\newblock `{K}nowable' as `known after an announcement'.
\newblock {\em Review of Symbolic Logic}, 1(3):305--334, 2008.

\bibitem{plaza:1989}
J.A. Plaza.
\newblock Logics of public communications.
\newblock In {\em Proc.\ of the 4th ISMIS}, pages 201--216. Oak Ridge National
  Laboratory, 1989.

\end{thebibliography}
\end{document}